\documentclass[10pt]{amsart}
\UseRawInputEncoding
\usepackage{amsfonts}
\usepackage{amssymb}
\usepackage{amsmath,amsthm}
\usepackage{amscd}
\usepackage{graphics}
\usepackage{amsmath,amssymb}
\usepackage{pdfsync}
\usepackage{mathrsfs}
\usepackage{enumitem}
\usepackage{stmaryrd}
\usepackage[colorlinks,linkcolor=blue,citecolor=blue,urlcolor=blue]{hyperref}
\usepackage{graphics}

\newcommand{\al}{\alpha}

\newcommand{\si}{\sigma}

\newcommand{\vp}{\varphi}

\newcommand{\Ga}{\Gamma}

\newcommand{\mL}{\mathcal{L}}

\newcommand{\cA}{{\mathcal A}}
\newcommand{\cB}{{\mathcal B}}

\newcommand{\cH}{{\mathcal H}}

\renewcommand{\leq}{\leqslant}
\renewcommand{\geq}{\geqslant}
\renewcommand{\le}{\leqslant}
\renewcommand{\ge}{\geqslant}

\newcommand{\e}{\mathrm e}

\newcommand{\abs}[1]{\lvert #1\rvert}
\newcommand{\nor}[1]{\lVert #1\rVert}

\newcommand{\bplus}[1]{\{\!#1\hspace{-1pt}\}_+}
\newcommand{\bminus}[1]{\{\!#1\hspace{-1pt}\}_-}
\newcommand{\tr}{\operatorname{tr}}

\newcommand{\beq}{\begin{equation}}
\newcommand{\eeq}{\end{equation}}
\newcommand{\bea}{\begin{eqnarray}}
\newcommand{\eea}{\end{eqnarray}}
\newcommand{\beas}{\begin{eqnarray*}}
\newcommand{\eeas}{\end{eqnarray*}}
\newcommand{\nn}{\nonumber}
\newcommand{\noi}{\noindent}
\newcommand{\tbf}{\textbf}

\newcommand{\bepm}{\begin{pmatrix} }
\newcommand{\epm}{\end{pmatrix}}

\newcommand{\mss}{\kern 1pt}

\newtheorem{proposition}{Proposition}
\newtheorem{theorem}{Theorem}

\newtheorem{lemma}{Lemma}
\newtheorem{definition}{Definition}
\theoremstyle{definition}

\newtheorem{remark}{\textbf{Remark}}

\begin{document}

\author{Jose Carrasco}
\address{Institut f\"ur Theoretische Physik, Universit\"at Innsbruck, Innrain 52, 6020 Innsbruck, Austria}
\email{joseacar@ucm.es}

\author{Giuseppe Marmo} \address{Dipartimento di Fisica, Universit\`a
  Federico II and Sezione INFN di Napoli, Via Cintia, Napoli, Italy}
\email{marmo@na.infn.it}

\author{Piergiulio Tempesta} \address{Instituto de Ciencias
  Matem\'aticas, C/ Nicol\'as Cabrera, No 13--15, 28049 Madrid,
  Spain\\ and Departamento de F\'{\i}sica Te\'{o}rica, Facultad de
  Ciencias F\'{\i}sicas, Universidad Complutense de Madrid, 28040 --
  Madrid, Spain } \email{  p.tempesta@fis.ucm.es, piergiulio.tempesta@icmat.es
}

\title[Entanglement monotones and formal groups]{New computable entanglement monotones \\ from formal group theory}
\date{February 08, 2021}
\maketitle
\begin{abstract}
We present a mathematical construction of new quantum information measures that generalize the notion of logarithmic negativity. Our approach is based on formal group theory. We shall prove that this family of generalized negativity functions, due their algebraic properties,  is suitable for studying entanglement in many-body systems.

Under mild hypotheses,  the new measures are computable entanglement monotones. Also, they are composable: their evaluation over tensor products can be entirely computed in terms of the evaluations over each factor, by means of a specific group law.

In principle, they might be useful to study separability and (in a future perspective) criticality of mixed states, complementing the role of R\'enyi's entanglement entropy in the discrimination of conformal sectors for pure states.

\end{abstract}
\tableofcontents


\section{Introduction}

The study of entanglement of many-body systems represents one of the
most relevant challenges of modern research in quantum physics, due to its intrinsic theoretical interest and its 
applicative relevance.  In this context, the determination of suitable information measures, allowing one to detect the entanglement in complex quantum systems is of outmost importance.
\cite{AFOVRMP2008}.

When analyzing compound systems made
up of spatially separated parties that can communicate with each
other, it is very common and natural to focus on protocols that consist on local operations assisted by classical communication (LOCC); they map the set of separable
states into itself. Operations preserving the positivity of the
density matrix after partial transposition (PPT-operations) are also of special
relevance, since all LOCC  are in particular PPT-protocols, which in turn map the set of states with positive partial transpose into
themselves.

In order to design a quantum information measure $E$ as an entanglement detector, certain conditions should be satisfied regarding LOCC or PPT-protocols. A fundamental requirement is that $E$ should be \textit{non-increasing on average}  under LOCC or PPT-preserving operations
\cite{PV1998,MHorodecki2001,PRHorodecki2001}. Precisely, we require
that
\begin{equation}\label{eq:mon}
  E(\rho)\ge\sum_i p_i E(\rho_i)
\end{equation}
where each of the states $\rho_i$, in a given Hilbert space, is obtained with probability $p_i$
when some LOCC or PPT-operation is applied to $\rho$. Another desirable property is that $E$ should be able to discriminate whether a state is separable or not, or at least, whether
a state belongs to the set of  states invariant under PPT-operations (PPT-invariant). The last condition is weaker: indeed, this set contains the set of separable states; however, there are PPT-invariant states that are not separable (they are said to contain bound entanglement).

An \textit{entanglement monotone} is a quantum measure satisfying both properties, namely it is non-increasing on average and it can discriminate the set of PPT-invariant or separable states.

The construction of quantum measures which fulfill these requirements and admit an explicit, analytic form is an interesting problem. Important examples are provided by the negativity and the
logarithmic negativity, introduced in the seminal
paper~\cite{VW2002}. In particular, in~\cite{Plenio2005} it was shown that the logarithmic negativity is an entanglement monotone.

In recent years, both the negativity and logarithmic negativity have
been largely investigated as entanglement
measures for mixed states~\cite{Plenio2005,PV2007}, as well as for their prominent role in several contexts of quantum field theory, in particular in the case of conformal field theories (CFTs) ~\cite{EZPRB2016,RACPRB2016,CCTPRL2012,CCTJPA2014}.

As is well known, both R\'enyi and Tsallis entropies can detect
criticality in some specific contexts when the von Neumann--Shannon entropy can not. 


There are one-dimensional quantum systems whose von
Neu\-mann--Sha\-nnon entanglement entropy coincides with that of a CFT
but their R\'enyi entanglement entropies do not (for all values of $n$). Consequently, instead of comparing their spectra (which would be the only definitive way of asserting that the quantum critical system is effectively described by a CFT) one can compute their R\'enyi entanglement entropies as a thinner criterion than the von Neumann--Shannon one (see also \cite{CFGTSR2017} for the relevance of R\'enyi's entropy in the study of multi-block entanglement entropy of free fermion systems). Since the standard (non-parametric) negativity has been recently computed for CFTs,  our ultimate goal is to find a meaningful and computable parametric generalization of the negativity in order to provide eventually thinner and more specific criteria to classify universality classes of one-dimensional quantum critical systems.

As a first step in this research program, the aim of this article is essentially of a mathematical nature, namely to establish a general mathematical framework that allows us to construct a new, large class of parametric quantum information measures playing the role of entanglement monotones (or quasi-monotones, as we will prove) for mixed states. More precisely, we will show that
a wide class of generalized entropic information functions can be
defined by means of formal group theory. Due to the fact that these
new group-theoretical functions widely generalize the notion of logarithmic negativity, we shall call them \textit{group negativities}.

As we will show, from a technical point of view,  group negativities are multi-parametric concave functions (generalized logarithms), depending on the $p$-norm of the partial transposition of a quantum state. The trace-norm subclass is recovered when $p=1$ and is related to quantized versions of the group entropies proposed
in~\cite{PT2016PRA}, computed over partially transposed states. Moreover, if the generalized logarithm is chosen to be the standard one, we recover the original logarithmic negativity introduced in~\cite{VW2002}.

More specifically, the \textit{trace-norm subclass} is composed by multi-parametric computable measures of entanglement. Indeed,  as a consequence of Proposition~\ref{prop:2} and under the  hypotheses of validity of the Peres criterion \cite{PeresPRL1996}, any trace-norm group negativity allows one to detect entanglement in mixed bipartite states: the \textit{strict positivity} of the functional is sufficient
to ensure that the state is entangled. We remind that Peres's criterion (positivity of the partial transpose of a state) is a necessary and sufficient condition for the separability of $2\times 2$ and $2\times 3$ systems, and is still necessary in higher dimensions~\cite{HHHPLA1996}. For the trace-norm subclass of group negativities, our main results are the following:
\begin{enumerate}
\item They are \textit{computable} measures of entanglement and
  provide separability tests for bipartite mixed states.
\item They are entanglement monotones.
\item They are \textit{composable}: each group negativity can be computed for a pure separable state in
  terms of the group negativities of each of its (non necessarily
  pure) reductions.
\end{enumerate}

The composability property is guaranteed by the specific functional
form of group negativities and could be important in the context of
distillability. We remind that if we have multiple copies of a bipartite
state $\rho$, its distillation rate is the best ratio between the
number of maximally entangled pairs which can be obtained from it
(distilled) by means of some LOCC and the number of copies of
the original state needed. The group negativity of the $n$ copies of
$\rho$ can be expressed through the group negativity of $\rho$.

Besides, as we shall see, generalized negativities constructed with the standard logarithm and $p$-norms with $p>1$ (\textit{p-norm group negativities}) are potentially interesting as well, since they represent auxiliary information tools that could be useful to determine bounds to distillability rates for distillation processes, in different scenarios. We observe that the simplest $p$-norm group
negativity (see Eq.~\eqref{eq:paddneg} below) is associated to the
additive formal group law (that is to say, the generalized logarithm is
the standard one) for all $p\ge1$ as in Definition~\ref{def:logpnorm};
we shall prove in Theorem~\ref{th:distill} that they provide upper
bounds for the entropy of distillation.

Concerning genuine $p$-norm group negativities, our main result, obtained under mild conditions, is that they are \textit{quasi-monotones} in the
sense that their increasing after an LOCC (or trace-preserving CP-PPT) operation is
bounded on average by a non-negative function $k(p)$ independent of the state
considered. This function can be made arbitrarily small, in the limit $p\to 1$.


From a mathematical point of view, the construction of group
negativities relies on the theory of formal
groups~\cite{Bochner1946,Haze}, which represents an important branch
of algebraic topology, with many applications in combinatorics and
number theory (see e.g. \cite{PT2010ASN,PT2015TRAN}). According to the arguments exposed above, we consider the composability property essential in order to discuss distillability, since the underlying formal group law controls how the
available information is redistributed when independent subsystems are combined into a new one. 

The article is organized as follows. In Section~\ref{Sec2}, we shall review briefly the basic aspects of formal group theory, with special emphasis on its role in the theory of generalized entropies.  In Section~\ref{NPPT}, the notion of negativity is discussed. In Section~\ref{Sec3}, the main definitions of the class of $p$-norm negativity  are introduced and their main properties are
proved. The  subclass of the trace-norm group negativities is discussed in Section~\ref{TNGN}. An additive $p$-norm is introduced in Section~\ref{Sec4}. Some open problems and future
perspectives are discussed in the final Section~\ref{Sec5}. 

\section{Groups and entropies: a general approach}\label{Sec2}
We shall start by reviewing  some aspects of the group-theoretical
classification of generalized entropies; also, we shall describe how this approach can be used in
our formulation of generalized negativities. Some definitions of  formal group theory will be presented in the Appendix (see also~\cite{Haze} for a thorough exposition, and~\cite{Serre1992} for a
shorter introduction to the topic).

\subsection{The Composability Axiom}

The notion of composability, introduced in \cite{Tsallis2009}, has been put in axiomatic form in~\cite{PT2016AOP}, \cite{PT2016PRA}, \cite{TJ2020SR} and related to formal group theory via the notion of group entropy.
We shall briefly discuss these concepts of composability and group entropies as in~\cite{PT2016AOP}, \cite{TJ2020SR} in order to illustrate the potential relevance
of the group-theoretical machinery in the study of composite quantum systems. We also mention that the relevance of the notion of group entropy in classical information geometry has been elucidated in \cite{RRT2019PRA}.

Let $\{p_i\}_{i=1,\cdots,W}$, with $W> 1$, $\sum_{i=1}^{W}p_i=1$ be a discrete probability distribution; we denote by $\mathcal{P}_{W}$ the set of all discrete probability distributions with $W$ entries. Let $S$ be a real function defined on $\mathcal{P}_{W}$.

\begin{definition}\label{composab} We shall say that $S$ is strictly \textit{composable} if there exists a continuous function
  of two real variables $\Phi(x,y)$ such that the following properties
  are satisfied.
  \begin{enumerate}[label={(C\arabic*)}]
  \item Composability: $S(P_A\times P_B)=\Phi(S(P_A),S(P_B))$, where $A$ and $B$ are two arbitrary statistically independent systems with associated probability distributions $P_A$ and $P_B$, respectively. \label{con:com}\\

  \item Symmetry: $\Phi(x,y)=\Phi(y,x)$.\label{con:symm}\\

  \item Associativity: $\Phi(x,\Phi(y,z))=\Phi(\Phi(x,y),z)$\label{con:asoc}\\

  \item Null-composability: $\Phi(x,0)=x$\label{con:null}
  \end{enumerate}
\end{definition}

Observe that the mere existence of a function $\Phi(x,y)$ taking care of the composition
process as in $\ref{con:com}$ is necessary, but  not sufficient to ensure that a given entropy may be suitable for information-theoretical or thermodynamic purposes: this function must satisfy all the requirements above to be admissible. Indeed, in general the entropy of the system compounded by the subsystems $A$ and $B$
should not vary if we exchange the labels $A$ and $B$, thus justifying
condition~\ref{con:symm}. In the same vein, condition~\ref{con:asoc} guarantees the composability of more than two systems in an associative way, this
property being crucial to define a zeroth law.  Finally, condition~\ref{con:null} in our opinion is also necessary since if we compound two systems $A$ and $B$ and the latter has zero entropy, then the total entropy must coincide with that of the former.

The set of requirements \ref{con:symm}--\ref{con:null} altogether represent the \textit{composability axiom}, which replaces the additivity axiom in the set of the four Shannon-Khinchin axioms.
These axioms, introduced by Shannon and Khinchin as conditions for an uniqueness theorem for the Boltzmann entropy, represent fundamental, non-negotiable requirements that an entropy $S[p]$ should satisfy to be physically meaningful: continuity with respect to all variables $p_1,\ldots,p_W$, maximization over the uniform distribution, expansibility (adding an event of zero probability does not affect the value of $S[p]$).

\begin{definition} \label{def:groupentropy}
A group entropy is a function $S: P_{W}\to \mathbb{R}^{+} \cup \{0 \}$ which satisfies the  Shannon-Khinchin axioms (SK1)-(SK3) and the composability axiom (C1)-(C4).
\end{definition}

Our construction of group negativities is inspired by this notion.

A function $\Phi(x,y)$ satisfying the properties \ref{con:symm}--\ref{con:null} is a formal group law. This is the origin of the connection between entropic measures and formal group theory (presented for the first time in \cite{PT2011PRE}),  as we shall illustrate in the forthcoming considerations. Among the examples of group entropies, we mention the class of $Z$-entropies, considered in \cite{PT2016PRA}.


We will show that the notion
of logarithmic negativity can be generalized by means of a
mathematical formalism based on formal group theory. Our main result is the
following: there exists a ``tower'' of new, parametric information measures, each of them reducing to the logarithmic negativity in a certain regime. 

Many other approaches to generalized entropies have been proposed in the literature. In particular, we mention the one developed in \cite{IS2014PHYSA} (strongly pseudo-additive entropies) and the one advocated in \cite{JK2019PRL}, based on statistical inference theory. These two approaches as well as the group-theoretical allow to define from an independent perspective  classes of entropies possessing many interesting properties.
%

\subsection{Group logarithms}
Our construction relies on the notion of group logarithm associated to every formal group law. There is a certain freedom concerning the regularity properties in its definition, depending on the  application under consideration. The standard logarithm is associated to the additive formal group law. 
\begin{definition}
 A group logarithm is a strictly increasing and strictly concave function $\log_{G}: (0,\infty)\to \mathbb{R}$, with $\log_{G}(1)=0$ (possibly depending on a set of real parameters), satisfying a functional equation of the form
\beq
\log_{G}(xy)= \chi(\log_{G}(x),\log_{G}(y)) \label{glog}
\eeq
where 
$\chi(x,y)$ fulfills the requirements $(C2)$-$(C4)$. This equation will be called the group law associated with $\log_{G}(\cdot)$.
\end{definition}


An useful result is the following, simple proposition.

\begin{proposition}
Let $G: \mathbb{R} \to \mathbb{R}$ be a strictly increasing function, vanishing at zero. The function $\Lambda_{G}(x)$ defined by
\beq
\Lambda_{G}(x):= G\left(\ln x^{\gamma}\right), \qquad x>0, \qquad \gamma> 0 \label{Glog}
\eeq
satisfies  a functional equation of the form \eqref{glog}.
\end{proposition}
\begin{proof}
Observe that
\begin{eqnarray*}
\Lambda_{G}(xy)&=& G\left(\ln x^{\gamma}+\ln y^{\gamma}\right)=G\left(G^{-1}(\Lambda_{G}(x))+G^{-1}(\Lambda_{G}(y))\right)\\
\nn &=& \chi(\Lambda_{G}(x),\Lambda_{G}(y)).
\end{eqnarray*} 
\noi In other words, we have
\beq
\chi(x,y)=G(G^{-1}(x)+G^{-1}(y)), \label{GPsi}
\eeq
which is indeed a group law.
\end{proof}
We shall consider functions of the form $G(t)= t+ O(t^2)$ for $t\to 0$, in such a way that $\chi(x,y)=x+y+\ldots$.
\begin{remark}
\noi From now on, we shall focus on group logarithms of the form
\beq\label{eq:GL}
\log_{G}(x)= G(\ln x) \ .
\eeq
Here $G$ is a suitable, strictly increasing function (vanishing at $0$), assuring concavity of $\log_{G}(x)$ (
see \cite{PT2016PRA} for a general discussion of sufficient conditions for $G$ to comply with these requirements).  

The formal inverse of a group logarithm will be called the associated group exponential; it is defined by
\beq
\exp_{G}(x)= e^{G^{-1}(x)}. \label{Gexp}
\eeq
\end{remark}
\begin{remark} A first, relevant example of nontrivial group logarithm is given by
the so called $q$-logarithm. We have
\beq \label{qlog}
 G(t)=\frac{\e^{(1-q)t}-1}{1-q},\qquad\log_q(x)=G(\ln
  x)=\frac{x^{1-q}-1}{1-q}, \qquad q>0 \,.
  \eeq
\noi This logarithm has
been largely investigated in connection with nonextensive statistics \cite{Tsallis1988}, \cite{Tsallis2009}.

Concerning group exponentials, notice that when $G(t)=t$, we obtain the standard exponential; when as before $G(t)=\frac{e^{(1-q)t}-1}{1-q}$, we recover the $q$-exponential $e_{q}(x)=\left[1+(1-q)x\right]^{\frac{1}{1-q}}_{+}$, and so on.
\end{remark}
Infinitely many other examples of group logarithms and exponentials are provided, for instance, in \cite{PT2011PRE}.

\subsection{Group entropies}
The general family of group entropies is usually defined in terms of generalized logarithms, which in turn allow us to realize the composition laws for the entropies in terms of formal group laws. We remind that a particular instance of the family of group entropies is the class of $Z$-entropies defined in~\cite{PT2016PRA} within the approach described above. Their
general form, for $\al >0$, is
\begin{equation}\label{eq:Z}
  Z_{G,\alpha}(p_1,\ldots,p_W):=\frac{\log_G\Big(\sum_{i=1}^{W}p_i^\al\Big)}{1-\al} \ ,
\end{equation} 
where   $\log_G$ is a group logarithm. 
We also mention that a widely generalized, relative-entropy version of this class, the  $(h,f)$-\textit{divergence}, has been introduced in \cite{RRT2019PRA}, in the context of information theory. Indeed, given two probability distributions $P,Q$ such that $p_i, q_i>0$, $i=1, \ldots, W$, we can immediately deduce the \textit{relative $Z$-class}:
\begin{equation}\label{eq:relGE}
  Z_{G,\alpha}(P || Q):=\log_G\Bigg[\Bigg(\sum_{i=1}^{W}p_i^{\al} q_{i}^{1-\alpha}\Bigg)^{\frac{1}{\alpha-1}}\Bigg] \ ,
\end{equation} 
$\alpha\neq 1$, which is nothing but a special version of the $(h,f)$-divergence, for $h(x)= G(\ln x^{\gamma})$,  $f(x)= x^{\alpha}$, $\gamma= \frac{1}{\alpha-1}$. Clearly, alternative classes are easily obtainable from the $(h,f)$-relative entropies of \cite{RRT2019PRA} by means of different choices of $h$ and $f$. 
To conclude this review, we recall that more general classes of group entropies have been defined in \cite{RRT2019PRA}, \cite{T2020CHA} and  \cite{TJ2020SR}. 

\vspace{3mm}

In this section, the group-theoretical formalism necessary for the formulation of our main results has been established. From a conceptual point of view, the main motivation for introducing this approach relies in the following observation: As the logarithmic negativity is associated to the additive formal group law, so it is natural to define new  entanglement measures associated to different, nonadditive composition laws by means of formal group theory.

\section{Negativity and PPT operations} \label{NPPT}

We introduce now the quantum formalism relevant in the forthcoming discussion. Let us denote by $\cB_1$ and $\cB_2$ the space of bounded linear
operators of the Hilbert spaces $\cH_1$ and $\cH_2$ respectively. For
a bipartite mixed state $\rho\in\cB_1\otimes\cB_2$, let us denote by
$\rho^\Ga$ its partial transposition with respect to $\cH_2$ (the
final result will not change if we choose $\cH_1$ in this
definition). The action of partial transposing is defined in the space
$\cB$ of bounded linear operators of the Hilbert space
$\cH=\cH_1\otimes\cH_2$ by extending (by linearity) the action
over pure separable states $\si\otimes\tau\in\cB$, with $\si\in\cB_1$ and $\tau\in\cB_2$:
\[
  (\si\otimes\tau)^\Ga=\si\otimes\tau^{\text{\tiny T}}
\]
where $\tau^{\text{\tiny T}}\in\cB_2$ is the transpose of $\tau$.

\begin{definition}\label{def:abs}
  Given an element $A\in\cB$, we introduce $\abs{A}:=\bplus{A}-\bminus{A}$ where $\{\cdot\}_+$ and $\{\cdot\}_-$ are its positive and negative parts, i.e., its restrictions to the eigenspaces of positive and negative eigenvalues respectively. The trace-norm $\nor{\cdot}_1$ of an operator $A$ is defined as $\nor{A}_1=\tr\abs{A}$.
\end{definition}

Note that $A=\bplus{A}+\bminus{A}$; if $A$ is Hermitian, then $\bplus{A}-\bminus{A}=\sqrt{AA^\dagger}$ where $\sqrt{B}$ represents any operator
$C\in\cB$ such that $C^2=B\in\cB$.

\begin{definition}
  Given a bipartite mixed state $\rho$, its negativity is defined to be the function $N(\rho):= \frac12(\nor{\rho^{\Ga}}_1-1)$, while its logarithmic negativity  is the function $L(\rho):=\ln\nor{\rho^{\Ga}}_1$.
\end{definition}

The monotonicity of $L(\rho)$ was proved in~\cite{Plenio2005}. Precisely, the inequality 
\[
  L(\rho)\ge\sum p_i L(\rho_i) 
\]
holds, where $\rho_i\propto\cA_i(\rho)$ is the normalized state associated
to outcome $i$ after applying the trace-preserving completely positive
 operation $\cA=\sum_i\cA_i$. Note that $\cA$ maps the set of
PPT states into itself and also that the result of applying $\cA$ to
$\rho$ can be seen as an ensemble with elements $\rho_i$ appearing
with probabilities $p_i=\tr\cA_i(\rho)$. The logarithmic negativity
$L$ is also an upper bound to distillable entanglement, as was shown
in \cite{VW2002}.

%

\section{$p$-norm group negativities}\label{Sec3}

\subsection{Definitions}
As before, we shall consider a composite quantum system, whose
associated Hilbert space $\cH$ has
dimension $N$; we shall denote by $\cB(\cH)$ the linear space of
bounded linear operators on $\cH$.

Consider the Scatten $p$-norms
\[
  \nor{A}_p=\big((s_1(A))^p+\cdots+(s_N(A))^p\big)^{1/p}, \qquad p\geq 1
\]
for any $A\in\cB$ with singular values $s_i(A)$; the limit $p\to\infty$ will be denoted by $\nor{\cdot}_\infty$.
We introduce now the main objects of our analysis.

\begin{definition}
The function $\mL_{G,p}:\cB\to\mathbf{R}$, defined for any state
  $\rho\in\cB$ and $p\geq 1$ as
\beq \label{eq:pgenneg}
\mL_{G,p}(\rho):=\log_G\nor{\rho^\Ga}_p
\eeq
 is said to be the $p$-norm  group  negativity of the state $\rho$. Here $\log_G(\cdot)$ is a group logarithm of the form \eqref{eq:GL}.
\end{definition}
Clearly, trace-norm group negativities are obtained when $p=1$.  
  A simple but interesting, new case is the additive one, obtained
when $\log_G(x)=\ln x$ for $p>1$.

\begin{definition} \label{def:logpnorm} The function
  \beq \label{eq:paddneg} \mathcal{L}_{p}(\rho)=\ln \nor{\rho^\Ga}_p, \qquad p\geq 1
  \eeq
  will be called the logarithmic p-norm negativity of a mixed
  state $\rho$.
\end{definition}
\begin{remark} Obviously, in information-theoretical applications one could replace  $\ln(x)$ with $\log_2(x)$ in
  Eq. \eqref{eq:paddneg} (as in the standard definition for $p=1$, by Vidal and Werner in \cite{VW2002}), without altering the
  main properties of the function.
\end{remark}
We will show that the quantity $\mL_{G,p}(\cdot)$ is bounded on
average under LOCC (trace-preserving CP-PPT) operations. This bound can be made arbitrarily
close to zero (and in particular is exactly zero in the limit
$p=1$).

We shall first deal with a deterministic trace-preserving CP-PPT
operation $\cA$; then we shall consider a general, not necessarily
deterministic operation, which maps a state $\rho$ into an ensemble of
states $\rho_i=\cA_i(\rho)$, each appearing with probability
$p_i=\tr\cA_i(\rho)$ (where each operation $\cA_i$ is a CP-PPT operation and
$\sum_i\cA_i$ is trace-preserving). Let us regard the partial
transposition as an operator $\Ga:\cB\to\cB$, with
$\Ga(\rho)=\rho^\Ga$. This operation is clearly involutive. We can define a linear map $\cA^\Ga:\cB\to\cB$ as
$\cA^\Ga(\si):=\Ga\circ\cA\circ\Ga(\si)$; equivalently, $\cA^\Ga\circ\Ga=\Ga\circ\cA$.

We propose a simple Lemma, useful in the forthcoming discussion.
\begin{lemma}\label{lem:PsiGaPos}
  Let $\cA:\cB\to\cB$ be any PPT quantum operation. If $\cA$ is
  positive and preserves the positivity of the partial
  transpose, then its partial transpose $\cA^\Ga$ is also positive.
\end{lemma}
\begin{proof}
  Since $\cA$ is a PPT operation, $\Ga\circ\cA(\rho)$ is  positive if $\Ga(\rho)$
  is positive.  Due to the fact that
  $\Ga$ is an involution, writing $\si=\Ga(\rho)$ and $\rho=\Ga(\si)$ we conclude that
  $\Ga\circ\cA\circ\Ga(\si)$ is  positive if $\Ga\circ\Ga(\si)=\si$ is
 positive. The result follows from the relation  $\cA^\Ga=\Ga\circ\cA\circ\Ga$.
\end{proof}

The following statement is due to Plenio~\cite{Plenio2005}.

\begin{lemma}\label{lem:Plenio}
  Let $\cA:\cB\to\cB$ be a trace-preserving completely positive operation. Then
  $\tr\abs{\cA(\rho)}\le\tr\abs{\rho}$.
\end{lemma}
\begin{proof}
  Note that
  \begin{equation}\label{eq:ineq}
    \bplus{\mss\cA(\cdot)\mss}=\bplus{\mss\cA(\{\cdot\}_+)+\cA(\{\cdot\}_-)\mss}\le\cA(\{\cdot\}_+) \ ,
  \end{equation}
  due to the fact that $\cA(\{\cdot\}_-)=-\cA(-\{\cdot\}_-)$ and $-\{\cdot\}_-$ is
  positive or zero, so that by linearity $\cA(\{\cdot\}_-)$ is
  negative or zero. Also, observe that
  \begin{multline*}
    \abs{\cA(\rho)}=\{\cA(\rho)\}_+-\{\cA(\rho)\}_-=\{\cA(\rho)\}_++\{-\cA(\rho)\}_+\\
    =\{\cA(\rho)\}_++\{\cA(-\rho)\}_+\le
    \cA(\rho_+)+\cA(\{-\rho\}_+) \ .
  \end{multline*}
Here we have used  twice inequality~\eqref{eq:ineq}. The result follows by noting
  that $\rho_-=-\{-\rho\}_+$ and taking into account linearity in the RHS of the
  last inequality to obtain $\abs{\rho}=\rho_+-\rho_-$ as the argument  of $\cA$.
\end{proof}

\begin{definition}
A group logarithm $\log_{G}(x)$ such that
\beq \label{eq:subadd}
\log_G(xy)\le\log_Gx+\log_G y
\eeq
will be said to be subadditive.
\end{definition}
In order to prove the main results of this section, we state the following
\begin{lemma} \label{lemma1} The inequality
  \begin{equation}
    \sum_i p_i\nor{\rho_i^\Ga}_p\le\tr\abs{\rho^\Ga}
  \end{equation}
holds for $p\geq 1$ under trace-preserving CP-PPT operations.
\end{lemma}
\begin{proof}
  Observe that
  \[
    \sum_ip_i\nor{\rho_i^\Ga}_p=\sum_i\nor{\cA_i^\Ga(\rho^\Ga)}_p=\sum_i(\tr\abs{\cA_i^\Ga(\rho^\Ga)}^p)^{1/p}
    \ ,
  \]
  where $\abs{\cA_i^\Ga(\rho^\Ga)}$ is of course a  positive operator.   Also, for any positive operator $A$ we have that
  \[
    \tr A^p\le(\tr A)^p, \qquad p\geq 1 \ .
  \]
  Thus, we obtain
  \begin{equation}
    \left(\tr\abs{\cA_i^\Ga(\rho^\Ga)}^p\right)^{1/p}\le\tr\abs{\cA_i^\Ga(\rho^\Ga)}\le\tr\cA_i^\Ga(\abs{\rho^\Ga}) \ ,
  \end{equation}
  where the last inequality follows from Lemma \ref{lem:Plenio}. Therefore we have shown that
  \begin{equation}
    \sum_i p_i\nor{\rho_i^\Ga}_p\le\sum_i\tr\cA_i^\Ga(\abs{\rho^\Ga})=\tr\abs{\rho^\Ga} \ .
  \end{equation}
\end{proof}
\subsection{Main result}
\begin{theorem}\label{th:MR}
  The group p-norm negativity
  $\mL_{G,p}(\rho)=\log_G\nor{\rho^\Ga}_p$ associated with a
  subadditive group logarithm, for any $p\geq 1$ is {\em bounded on average} under trace-preserving CP-PPT operations, that is,
  there exists a constant $k(p)$  such that \beq \label{eq:pnormineq} \sum_i p_i
  \mL_{G,p}(\rho_i)-\mL_{G,p}(\rho)\le k(p) \ .  \eeq
\end{theorem}
\begin{proof} Since by definition a group logarithm is a concave
  function, then
  \begin{equation}\label{eqthm1}
    \sum_ip_i\, \mL_{G,p}(\rho_i)=\sum_i p_i\log_G\nor{\rho_i^\Ga}_p\le\log_G\Big(\sum_i p_i\nor{\rho_i^\Ga}_p\Big)\,.
  \end{equation}
  Now, since $\log_G(x)$ is also a strictly increasing function, by means of Lemma
  \ref{lemma1}, we deduce
  \[
    \sum_i p_i \mL_{G,p}(\rho_i)\le\log_G\tr\abs{\rho^\Ga}\,.
  \]
  Let us denote by $N$ the dimension of the ambient Hilbert space $\cH$. We
  can observe that
  \[
    \tr\abs{\rho^\Ga}=\nor{\rho^{\Ga}}_1\le c(p)\nor{\rho^\Ga}_p \ ,
  \]
  where $c(p)=N^{1-1/p}$, and $\rho\in\cB(\cH)$.

\noi Finally, due to subadditivity of $\log_{G}$,  we have that
  \[
    \sum_ip_i
    \mL_{G,p}(\rho_i)\le\log_G\nor{\rho^\Ga}_1\le\log_Gc(p)+\log_G\nor{\rho^\Ga}_p\,.
  \]
To conclude,  we introduce \beq \label{kdep} k(p):=\log_G c(p). \eeq Thus, the previous
  inequality reduces to relation \eqref{eq:pnormineq}.
\end{proof}
\begin{remark}
 An interesting aspect of inequality \eqref{eq:pnormineq} is that
  $k(p)$ can be made arbitrarily small by
  considering $p$-norms where $p=1+\delta$, with $\delta>0$ arbitrarily close to zero.
Thus, for any $\epsilon >0 $ there exists a value  $\delta$ such that 
\beq \label{eq:qm}
\sum_i p_i \mL_{G,p}(\rho_i)\le \mL_{G,p}(\rho) + \epsilon \ .
\eeq
Due to the latter property, we shall say that the $p$-norm negativity $\mL_{G,p}(\rho)$ is an \textit{$\epsilon$-monotone} (or
  \textit{quasi-monotone}).
\end{remark}

\begin{remark} The hypotheses of Theorem \ref{th:MR} are actually satisfied by an infinite family of group
  logarithms. For instance, the  group  exponential and logarithm considered in Remark 2 possess all required
  properties: Indeed, $\log_q(x)$ is strictly concave, monotonically
  increasing and subadditive for $q>1$. In general, we have
  $ \log_q(\rho^{X} \otimes
  \rho^{Y})=\log_q(\rho^{X})+\log_q(\rho^{Y})+(1-q)\log_q(\rho^{X})\log_q(\rho^{Y})
  \ .  $ We introduce now the corresponding negativity measure.
  \begin{definition}  The \textit{$p$-norm $q$-negativity} for any $\rho\in\cB$ is the function
  \beq \label{eq:pnormqneg} \mathcal{L}_{p}^{(q)}(\rho):=\frac{
    (\parallel \rho^{\Gamma}\parallel_{p})^{1-q}-1}{1-q}, \qquad q>1 \ .  \eeq
\end{definition}

\end{remark}

In summary, the class of subadditive $p$-norm group negativities, defined in this Section, share with the standard negativity the crucial property of being bounded on average, as assured by Theorem \ref{th:MR};  this result represents one of the main contributions of this work. 

\section{Trace-norm group negativities and monotonicity} \label{TNGN}
A particular case of the previous construction, interesting on its own, is the trace-norm class, corresponding to $p=1$ in the previous analysis.

\begin{definition} \label{def:trace-norm}
A trace-norm group negativity $L_G:\cB\to\mathbf{R}$, where $\cB$ is the space of
  bounded linear operators of a Hilbert space $\cH$,  is the function
  \begin{equation} \label{eq:tngn}
    L_{G}(\rho):=\log_{G} \parallel \rho^{\Gamma}\parallel_{1}, \qquad \rho\in\cB,
  \end{equation}
where $\log_{G}(x)$ is a group logarithm of the form \eqref{eq:GL}.
\end{definition}


As an immediate consequence of Theorem \ref{th:MR}, we have the
following result.
\begin{theorem}
  A trace-norm group negativity $L_{G}(\rho)=\log_G\nor{\rho^\Ga}_1$
  associated with a subadditive group logarithm satisfies the monotonicity property: \beq \label{eq:normineq} \sum_i p_i L_{G}(\rho_i)\le
  L_{G}(\rho) \ .  \eeq
\end{theorem}
\begin{proof}
  It suffices to assume $p=1$ in the previous discussion. In
  particular, we have identically $c(1)=1$ and $k(1)=0$ into
  Eq. \eqref{eq:pnormineq}.
\end{proof}

The trace-norm group negativities can be regarded as a new quantum
version of the $Z$-entropies
introduced in~\cite{PT2016PRA}. The main novelty of the present construction is that
the functional \eqref{eq:tngn} is the trace-norm of the partial transposition of a quantum
state whose spectrum need not be, in general, a probability
distribution. 
\begin{proposition}\label{prop:2} Any trace-norm group negativity is
  positive semi-definite over the space $\cB$ of bounded linear
  operators of a Hilbert space $\cH$ and vanishes for states with
  positive partial transpose. Furthermore, it is strictly composable: if
  $\cH=\cH_1\otimes\cH_2$ then
  \[
    L_G(\si\otimes\tau)=\Phi(L_G(\si),L_G(\tau))
  \]
  for any pair of states $\si\in\cB_1$ and $\tau\in\cB_2$ where
  $\Phi(x,y)=G(G^{-1}(x)+G^{-1}(y))$.
\end{proposition}
\begin{proof}
  Since $\tr\rho=\tr\rho^\Ga$, it follows that
  $\bminus{\mss\rho^\Ga\mss}=0$ only when $\nor{\rho^\Ga}_1=1$.
  Definition~\ref{def:trace-norm} implies that $L_G(\rho)=0$ in this case, whereas
  $L_G(\rho)>0$ in the other cases (being $\log_G$ strictly
  increasing). Composability is assured by the functional equation associated with the group logarithm $\log_{G}(x)=G(\ln x)$.
\end{proof}

In our framework, the original logarithmic negativity~\cite{VW2002} corresponds to the choice
$G(t)=t$ which leads to the additive group $\Phi(x,y)=x+y$. A new non-trivial example is provided by the use of the $q$-logarithm
of Eq.~\eqref{qlog}.

\begin{definition}\label{def:qneg} The trace-norm $q$-negativity of a
  state $\rho\in\cB$ is the function
  \begin{equation*}
    L^{(q)}(\rho):=\frac{\nor{\rho^\Ga}_1^{1-q}-1}{1-q}, \qquad q>1  \,.
  \end{equation*}
\end{definition}
For $q\to 1$, it reduces to the logarithmic negativity $\lim_{q\to 1}L^{(q)}(\rho)=L(\rho)$. Furthermore,
for any pair $\si\in\cB_1$ and $\tau\in\cB_2$ it follows
\[
  L^{(q)}(\si\otimes\tau)=L^{(q)}(\si)+L^{(q)}(\tau)+(1-q)L^{(q)}(\si)L^{(q)}(\tau)\,.
\]
The trace-norm $q$-negativity is thus associated to the same
composition law  (called in algebraic topology the multiplicative formal group law \cite{Haze}), namely $\Phi(x,y)=x+y+(1-q)x\,y$ as both the classical and the quantum versions of
the Tsallis entropy~\cite{Tsallis1988,Tsallis2009}. It is
worth mentioning that the entanglement entropy associated to the Tsallis
entropy (its evaluation over the reduced density matrix of a bipartite
pure state) has been used in~\cite{TLB2001} to characterize the
separability of a family of quantum states, correctly recovering Peres
criterion for a concrete family of states.

Actually, Definitions ~\ref{def:trace-norm} and \ref{def:qneg} are naturally adapted to the Peres criterion [24], which can be applied to any state. Whenever the Peres criterion is necessary and sufficient, then $L_{(q)}(\rho)>0$ and, generally speaking, $L_G(\rho)>0$ for entangled states only. However, in a more general context (i.e. for higher dimensional spaces), $L_G(\rho)>0$ for NPT-entangled states only (here, by NPT-entangled states we mean quantum states with non-positive partial transpose).


Clearly,
the trace-norm $q$-negativity suggests further generalization in terms
of general entanglement witnesses \cite{PV2007}, namely quantities that separate an
entangled state from the set of separable ones in more general
scenarios, whereas partial transpose separates in a necessary and sufficient way quantum systems
associated to Hilbert spaces of dimension strictly lower than eight,
namely when $\min_{i=1,2}\dim\cH_i=2$ and $\max_{i=1,2}\dim\cH_i=2,3$.

When $\dim\cH\ge 8$, positivity of partial transposition is only necessary for separability and thus there exist entangled states with positive partial transpose. For all of them, $\bminus{\mss\rho^\Ga\mss}=0$ and $L_G(\rho)=0$, in particular, when $G(t)=t$ one recovers the well-known fact that the logarithmic negativity vanishes $L(\rho)=0$ on PPT entangled states.

\begin{remark}
The fact that trace-norm group negativities  are strictly composable  is a non-trivial property, essentially related to their non-trace-form
functional expression. Indeed, when dealing with standard entropies over a probability space, classical strict
 composability prevents the use of infinitely many trace-form entropies, namely functions of probability distributions
 $(p_1,\ldots,p_W)$ of the form $
   \sum_{i=1}^{W}f(p_i)$, $f(0)=f(1)=0$.
Precisely, a theorem proved in \cite{ET2017} states, under mild  hypotheses, that the most general trace-form entropy which is
strictly composable is Tsallis entropy (recovering Boltzmann's entropy when $q\to 1$). Thus, using the more commonly adopted trace-form functionals one is lead to \textit{weakly} composable group entropies \cite{PT2016AOP}, which are composable  over the product of uniform distributions only. Instead, strictly composable entropies are
allowed in the non-trace-form functional class.   Indeed, each of the $Z$-entropies in Eq.~\eqref{eq:Z} is
strictly composable, with a specific composition law;
thus, one can associate a trace-norm group negativity with each of them. The first of such pairs
is represented by the original Tsallis $q$-entropy and the trace-norm
$q$-negativity of Definition~\ref{def:qneg}.
\end{remark}

%

To summarize, the group norms for $p=1$ possess more regularity properties (they are entanglement monotones) and can be used as kind of nonadditive entanglement witnesses. 

To complete the discussion,  we shall focus on the additive case for $p\geq 1$.

\section{p-norm additive negativity as an upper bound for distillability}\label{Sec4}
A crucial property of the additive $p$-norm \eqref{eq:paddneg}
introduced in the present work is the fact that it represents an upper bound for distillability. In our analysis, we shall closely follow the notation and the discussion of Ref. \cite{VW2002}.


We define
$\mathcal{L}_{p}(\Omega):=\ln \nor\Omega_p= \frac{1-p}{p} \ln N $, where
$\Omega$ is the (diagonal) density matrix of the maximally mixed
(separable) state. Then, we can introduce the \textit{normalized p-norm negativity}
\beq \label{eq:normpnorm} \tilde{\mL}_p(\rho):= \mathcal{L}_{p}(\rho)
-\mathcal{L}_{p}(\Omega) \ .  \eeq Note that the standard (trace-norm)
logarithmic negativity is already normalized:
\[
  \tilde{\mL}_1(\rho)=L(\rho) \ .
\]
In a completely analogous way, one can normalize any $p$-norm group negativity.

Assume that we have a bipartite state $\rho$ and multiple copies of it obtained by means of LOCC. We recall that its distillation rate is the best rate at which we can extract near-perfect singlet states from its copies. In particular, given a large number of copies of the state, its \textit{asymptotic} distillation rate is called its entanglement of distillation $E_{D}(\rho)$.

Let us consider $n_{\alpha}$ copies of $\rho$ and let $Y$ be a
maximally entangled state of two qubits. Then, we are interested in
the best approximation to $m_{\alpha}$ copies of $Y$ that can be
obtained from $\rho^{\otimes n_{\alpha}}$ by means of LOCC.

We introduce \cite{VW2002} 
\beq \label{eq:infimum} \Delta(Y^{\otimes
  m_{\alpha},\rho^{\otimes n_{\alpha}}})= \text{inf}_{\mathscr{P}} \parallel Y^{\otimes
  m_{\alpha}} - \mathscr{P}(\rho^{\otimes n_{\alpha}}) \parallel_1 \ .  
\eeq
Here $\mathscr{P}$ runs over all deterministic protocols obtained from LOCC.

We say that $c$ is an achievable distillation rate for $\rho$ , if for
any sequences $n_{\alpha}$, $m_{\alpha}\to \infty$ of integers such
that $\limsup_{\alpha}(n_{\alpha}/ m_{\alpha})\leq c$ we have \beq
\lim_{\alpha}~\Delta(Y^{\otimes m_{\alpha},\rho^{\otimes
    n_{\alpha}}})=0 \ .  \eeq Thus, the distillable entanglement is
the supremum of all achievable distillation rates. If we allow a small
error level, we can introduce the distillable entanglement at error
level $\epsilon$, denoted by $E_{D}^{\epsilon}(\rho)$, which is
characterized by the weaker condition \beq
\lim_{\alpha}~\Delta(Y^{\otimes m_{\alpha},\rho^{\otimes
    n_{\alpha}}})\leq \epsilon \ .  \eeq

In this context, our main result is the following

\begin{theorem} \label{th:distill} Let $\tilde{\mathcal{L}}_p(\rho)$
  be the normalized logarithmic p-norm negativity. Then, for any $p\geq 1$
  we have \beq \label{eq:dist} \tilde{\mL}_p(\rho)\geq
  E_{D}^{\epsilon} \ .  \eeq
\end{theorem}
\begin{proof}
  As is well known \cite{VW2002}, the standard logarithmic negativity
  satisfies the upper bound
  \beq\label{UB}
  L({\rho}) \geq E_{D}^{\epsilon} .
  \eeq
  We also remind the
  inequalities ($p>1$) \beq \label{eq:normineq}
  \parallel\cdot\parallel_p \hspace{1mm} \leq
  \parallel\cdot\parallel_1 \leq N^{1-1/p} \parallel\cdot\parallel_p \
  .  \eeq Consequently, from the first inequality \eqref{eq:normineq}, we get
  \beq \label{eq:fineq} \mathcal{L}_{p}(\rho)\leq L(\rho).  \eeq
  From the second one, we have
\[
  L{(\rho)}\leq -\mL_p(\Omega)+ \mL_p(\rho) \ .
\]
Thus, due to inequality \eqref{UB}, we get
\beq \label{eq:secineq}
\mathcal{L}_{p}(\rho) - \mL_p(\Omega) \geq E_{D}^{\epsilon} \ .
\eeq
Using Definition \ref{eq:normpnorm}, we conclude that
\beq
\tilde{\mL}_p(\rho)\geq E_{D}^{\epsilon} \ .
\eeq
\end{proof}
\section{Future Perspectives} \label{Sec5}
As we have shown, group theory offers a natural way to generalize the
notion of negativity. This work represents a first exploration of a
new, infinite class of easily computable entropic-type measures of entanglement. We have introduced the $p$-norm and trace-norm group negativities and focused on the mathematical study of some of the main analytical properties of the class. Our main result is the construction of a large family of composable entanglement monotones.

Several aspects of the theory deserve further analysis. It is clear
that composability is crucial in order to compute entanglement entropy
of bipartite or multipartite systems in a natural way, starting from
the knowledge of the entropy of its constituents.  As we suggested,
such a property is fundamental to study distillable
entanglement. Therefore, an interesting open problem is to ascertain
if all of the group-theoretical negativities introduced here, apart
the logarithmic $p$-norm negativity, can provide upper or lower bounds
to the asymptotic distillation rate by means of LOCC, when we consider
a large number of copies of the state $\rho^{\otimes n_{\alpha}}$.

At the same time, it would be very interesting to apply the large
family of entropic  functionals introduced in this work in the study of
finite temperature systems in conformal field theories
\cite{CCTJPA2014}. From this point of view, one-parametric (or
multi-parametric) entanglement monotones could play a role similar to
that played by R\'enyi's entropy in the case of the entanglement
detection of the ground state of one-dimensional many body systems, and
in the study of their criticality properties \cite{Calabrese2004}.

We wish to point out that the language of formal group theory can be
directly related to the study of alternative formulations of both
classical and quantum mechanics.  Indeed, as shown in \cite{EIMM2007},
the linear structure of the theory can be replaced by a non-additive
structure generated by means of a suitable diffeomorphism, which would
play the same role as the group logarithm of the present theory. In
particular, this perspective opens the possibility of performing
non-equivalent Weyl quantizations of physical systems, circumventing
the von Neumann uniqueness theorem. The generalized negativities
introduced in the present work could play a significant role in these
alternative formulations. We shall discuss these aspects in detail
elsewhere.

Another interesting problem is to give an interpretation of both group negativities  and the quantum version of the relative entropies \eqref{eq:relGE} within the context of quantum information geometry,
especially in connection with the problem of tomographic
reconstruction of quantum metrics \cite{Amari2001}, \cite{Amari2016},
\cite{MMVV2017}, \cite{CCLMMVV2018}. Also, an intriguing problem is to establish
a connection between the approach proposed in this article and that one developed
independently in \cite{WW2019}.

Finally, we also plan to apply generalized negativities to the study
of entanglement properties of some concrete examples of quantum
systems, in particular integrable spin chains of Haldane-Shastry type
\cite{Haldane1988} \cite{Shastry1988}.  

Work is in progress along these lines.

\section*{Acknowledgement}
The authors wish to thank the anonimous referees for useful suggestions, which improved the readability of the article. 

\vspace{2mm}

J.C. would like to thank Aleksander M. Kubicki for several enlightening discussions.  

G.M. would like to thank the support provided
by the Santander/UC3M Excellence Chair Programme 2019/2020. 

The research of P.T. has been supported by the research project
PGC2018-094898-B-I00, Ministerio de Ciencia, Innovaci\'{o}n y Universidades and Agencia Estatal de Investigaci\'on,
Spain, and by the Severo Ochoa Programme for Centres of Excellence in R\&D
(CEX2019-000904-S), Ministerio de Ciencia, Innovaci\'{o}n y Universidades y Agencia Estatal de Investigaci\'on, Spain. 

G. M and P. T. are members of the Gruppo Nazionale di Fisica Matematica (INDAM), Italy.

\begin{section}{Appendix:Formal groups and formal rings}\label{subsec:formal}

Let $R$ be a commutative associative ring with identity, and
$R\llbracket x_1,x_2,\ldots\rrbracket$ be the ring of formal power
series in the variables $x_1,x_2,\ldots$ with coefficients in $R$. We
shall assume that $R$ is torsion-free.

\begin{definition}{\cite{Bochner1946}}\label{def:fgl} A commutative one-dimensional
  formal group law over $R$ is a formal power series
  $\Phi\in R\llbracket x,y\rrbracket$ such that
  \begin{enumerate}
  \item $\Phi \left( x,0\right) =\Phi \left( 0,x\right) =x$
  \item
    $\Phi \left( \Phi \left( x,y\right) ,z\right) =\Phi \left( x,\Phi
      \left( y,z\right) \right).$
  \end{enumerate}
  The formal group law is said to be commutative if
  $\Phi(x,y)=\Phi(y,x)$.
\end{definition}

Observe that the existence of an inverse formal series
$\vp\in R\llbracket x\rrbracket$ such that $\Phi(x,\vp(x))=0$ is a direct
consequence of the previous definition. This justifies the ``group'' terminology for these algebraic structures.

Let $B=\mathbf{Z}\llbracket b_1,b_1,\ldots\rrbracket$ and consider the
following series in $B\llbracket s\rrbracket$
\begin{subequations}\label{eqs:FG}
  \begin{equation}\label{eq:F}
    F(s)=s+\sum_{i=1}^{\infty}b_i\frac{s^{i+1}}{i+1}\,.
  \end{equation}
If $G\in B\llbracket t\rrbracket$ is its compositional inverse (namely $F(G(t))=t$ and $G(F(s))=s$), one has
  \begin{equation}\label{I.2}
    G(t)=t+\sum_{k=1}^{\infty}a_k\frac{t^{k+1}}{k+1} \ ,
  \end{equation}
\end{subequations}
with $a_{1}=-b_1$, $a_2=\frac32b_1^2-b_2, \ldots$.
 Given the formal power series $F$ and $G$ as in Eqs.~\eqref{eqs:FG}, the Lazard formal
group law~\cite{Haze} is defined by the formal power series
\[
  \Phi_{\text{L}}(s_1,s_2)=G(G^{-1}(s_1) +G^{-1}(s_2))
\]
whose coefficients, generate
over $\mathbf{Z}$ a subring $L\subset B\otimes\mathbf{Q}$. In other words, the Lazard ring is defined over a subring of the original ring $B\otimes\mathbf{Q}$, called the
Lazard ring.

For any commutative one-dimensional formal group law over any ring $R$,
there exists a unique homomorphism $L\to R$, under which the Lazard
group law is mapped into the given group law. This is called the
universal property of the Lazard group. Also, it is important to notice that for any commutative one-dimensional formal group law $\Phi(x,y)$ over $R$, there exists a series $\phi(x)$ such that
\[
\phi(x)= x+ O(x^2), \quad \text{and} \quad \Phi(x,y)= \phi^{-1}\left(\phi(x)+\phi(y)\right).
\]
Finally, let us also define
the notion of formal ring (see~\cite{Carrasco2019}).

\begin{definition}\label{def:formalring}
  Let $\left(R,+,\cdot\right)$ be a unital ring. A formal ring is a
  triple $(R,\Phi,\Psi)$ where
$\Phi,\Psi\in R\llbracket x,y\rrbracket$ are formal power series such
that
\begin{enumerate}
\item $\Phi$ is a commutative formal group law.
\item $\Psi$ satisfies the relations  \begin{align*}
  \Psi(\Psi(x,y),z)&=\Psi(x,\Psi(y,z))\\
   \Psi(x,\Phi(y,z))&=\Phi(\Psi(x,y),\Psi(x,z))\\
    \Psi(\Phi(x,y),z)&=\Phi(\Psi(x,z),\Psi(y,z)).
 \end{align*}
\end{enumerate}
The formal ring will be said to be commutative if
$\Psi(x,y)=\Psi(y,x)$.
\end{definition}

\end{section}


\begin{thebibliography}{00}
%
\bibitem{Amari2001} S. I.  Amari, H. Nagaoka. \textit{Methods of
    information geometry}, Vol. 191, American Mathematical Society
  (2000).

\bibitem{Amari2016} S. I. Amari. Information geometry and its
  applications, Applied Mathematical Sciences. Japan: Springer (2016).

\bibitem{AFOVRMP2008} L. Amico, R. Fazio, A. Osterloh and
  V. Vedral. Entanglement in many-body systems,
  Rev. Mod. Phys. \textbf{80}, 517 (2008).

\bibitem{Bochner1946} S. Bochner. Formal Lie groups,
  Ann. Math. \textbf{47} (1946), 192--201.

\bibitem{Calabrese2004} P. Calabrese and J. Cardy. Entanglement
  entropy and quantum field theory, J. Stat. Mech. P06002 (2004),
%

\bibitem{CCTPRL2012} P Calabrese, J. Cardy and E. Tonni. Entanglement
  negativity in quantum field theory, Phys. Rev. Lett. \textbf{109}
  (13), 130502 (2012).

\bibitem{CCTJPA2014} P. Calabrese, J. Cardy and E. Tonni. Finite
  temperature entanglement negativity in conformal field theory,
  J. Phys. A: Math. Theor. \textbf{48} (1), 015006 (2014)

\bibitem{CFGTSR2017} J. Carrasco, F. Finkel, A. Gonzalez-Lopez, and P. Tempesta. A duality principle for the multi-block entanglement entropy of free fermion systems, Nature - Scient. Rep. \textbf{7} (1), 11206 (2017).

\bibitem{Carrasco2019} J. Carrasco, P. Tempesta. Formal rings, arxiv:
  1902.03665 (2019).
%

\bibitem{CCLMMVV2018} F. M. Ciaglia, F. Di Cosmo, M. Laudato,
  G. Marmo, F. M. Mele, F. Ventriglia, P. Vitale. A pedagogical
  intrinsic approach to relative entropies as potential functions of
  quantum metrics: the q-z Family, Ann. Phys. \textbf{395} (2018).

\bibitem{EZPRB2016} V. Eisler, Z. Zimborás. Entanglement negativity in  two-dimensional free lattice models, Phys. Rev. B \textbf{93} (11),  115148 (2016).

\bibitem{ET2017} A. Enciso, P. Tempesta.  Uniqueness and
  characterization theorems for generalized entropies,
  J. Stat. Mech. 123101 (2017).

\bibitem{EIMM2007} E. Ercolessi, A. Ibort, G. Marmo and
  G. Morandi. Alternative linear structures for classical and quantum
  systems, Int. J. Mod. Physics A \textbf{22}, 3039-3064 (2007).

\bibitem{Haldane1988} F.~D.~M. Haldane. {E}xact
  {J}astrow--{G}utzwiller resonating-valence-bond ground state of the
  spin-$1/2$ antiferromagnetic {H}eisenberg chain with $1/r^2$
  exchange.  Phys. Rev. Lett. \textbf{60} 635 (1988).

\bibitem{Shastry1988} B.~S. Shastry. {E}xact solution of an ${S}=1/2$
  {H}eisenberg antiferromagnetic chain with long-ranged interactions.
  Phys.  Rev. Lett. \textbf{60} 639 (1988).

\bibitem{Haze} M. Hazewinkel. \textit{Formal Groups and Applications},
  Academic Press, New York (1978).




\bibitem{MHorodecki2001} M. Horodecki. Entanglement measures, Quantum
  Inf. Comp. \textbf{1}, 3 (2001).

\bibitem{HHHPLA1996} M.  Horodecki, P. Horodecki,
  R. Horodecki. Separability of mixed states: necessary and sufficient
  conditions, Phys. Lett. A \textbf{223}, 1 (1996).

\bibitem{PRHorodecki2001} P. Horodecki and R. Horodecki.  Distillation
  and bound entanglement, Quantum Inf. Comp. \textbf{1}, 45 (2001).

\bibitem{IS2014PHYSA} V. Ilic and M. Stankovic, Generalized Shannon-Khinchin axioms and uniqueness theorem for pseudo-additive entropies, Physica A \textbf{411}, 138-145 (2014).


\bibitem{JK2019PRL} P. Jizba and J. Korbel, Maximum Entropy Principle in statistical inference: case for non-Shannonian entropies, Phys. Rev. Lett. \textbf{122}, 120601 (2019).


\bibitem{Khinchin} A. I. Khinchin. \textit{Mathematical Foundations of
    Information Theory}, Dover, New York (1957).

\bibitem{MMVV2017} V. I. Man'ko, G. Marmo, F. Ventriglia,
  P. Vitale. Metric on the space of quantum states from relative
  entropy. Tomographic reconstruction, J. Phys. A:
  Math. Theor. \textbf{50}, 335302 (2017).

\bibitem{PeresPRL1996} A. Peres. Separability Criterion for Density
  Matrices, Phys. Rev. Lett. \textbf{77}, 1413 (1996).

\bibitem{Plenio2005} M. B. Plenio. Logarithmic negativity: A full
  entanglement monotone that is not convex,
  Phys. Rev. Lett. \textbf{95}, 090503 (2005).

\bibitem{PV1998} M. B. Plenio and V. Vedral. Teleportation,
  entanglement and thermodynamics in the quantum world,
  Contemp. Phys. \textbf{39}, 431 (1998).

\bibitem{PV2007} M. B. Plenio and S. Virmani. An introduction to
  entanglement measures, Quantum Inf. Comp. \textbf{7}, 001 051
  (2007).

\bibitem{RRT2019PRA} M. A. Rodr\'iguez, A. Romaniega and P. Tempesta, A new class of entropic information measures, 
formal group theory and information geometry, Proc. Royal Soc. A \textbf{475}, 20180633 (2019).

\bibitem{RACPRB2016} P. Ruggiero, V. Alba and
  P. Calabrese. Entanglement negativity in random spin chains,
  Phys. Rev. B \textbf{94} (3), 035152 (2016)

\bibitem{Serre1992} J.--P. Serre. \textit{Lie algebras and Lie
    groups}, Lecture Notes in Mathematics, 1500 Springer--Verlag,
  1992.

\bibitem{Shannon} C. E. Shannon.  A mathematical theory of
  communication, Bell Syst. Tech. J. \textbf{27} (1948) 379--423,
  \textbf{27} 623--653 (1948).

\bibitem{Shannon2} C. E. Shannon and W. Weaver.  \textit{The
    mathematical Theory of Communication}, University of Illinois
  Press, Urbana, USA (1949).
%
%
%
%
%
%
%

\bibitem{PT2010ASN} P. Tempesta.  L--series and Hurwitz zeta functions
  associated with the universal formal group, Annali
  Sc. Norm. Superiore, Classe di Scienze, IX, 1--12 (2010).

\bibitem{PT2011PRE} P. Tempesta.  Group entropies, correlation laws
  and zeta functions, Phys. Rev. E \textbf{84}, 02112 (2011).

\bibitem{PT2015TRAN} P. Tempesta.  The Lazard formal group, universal
  congruences and special values of zeta functions,
  Trans. Am. Math. Society, \textbf{367}, 7015-7028 (2015).
%

\bibitem{PT2016AOP} P. Tempesta. Beyond the Shannon-Khinchin
  formulation: The composability axiom and the universal group
  entropy. Ann. Phys. \textbf{365}, 180--197 (2016).

\bibitem{PT2016PRA} P. Tempesta. Formal Groups and $Z$--Entropies,
  Proc. Royal Soc. A, Vol. 472, 20160143 (2016).

\bibitem{T2020CHA} P. Tempesta, Multivariate group entropies, super-exponentially growing complex systems, and functional equations, Chaos \tbf{30}, 123119 (2020).


\bibitem{TJ2020SR} P. Tempesta and H. J. Jensen, Universality Classes and
Information-Theoretic Measures of Complexity via Group Entropies, Nature - Sci. Rep. \textbf{10}, 5952 (2020).



\bibitem{Tsallis1988} C. Tsallis. Possible generalization of the
  Boltzmann--Gibbs statistics, J. Stat. Phys. \textbf{52}, Nos. 1/2,
  479--487 (1988).

\bibitem{Tsallis2009} C.  Tsallis. \textit{Introduction to
    Nonextensive Statistical Mechanics--Approaching a Complex World},
  Springer, Berlin (2009).

\bibitem{TLB2001} C. Tsallis, S. Lloyd and M. Baranger. Peres
  criterion for separability through nonextensive entropy,
  Phys. Rev. A \textbf{63}, 042104.
%
%
%

\bibitem{VW2002} G. Vidal and R. F. Werner. Computable measures of
  entanglement, Phys. Rev. A \textbf{65}, 132314 (2002).

\bibitem{WW2019} X. Wang and M. Wilde. $\alpha$-logaritmic negativity, Phys. Rev. A \textbf{102}, 032416 (2020).
\end{thebibliography}
\end{document}